\documentclass[11pt,a4paper]{article}
\usepackage[margin=1in]{geometry}
\usepackage{amsthm}  
\usepackage{amsmath} 
\usepackage{amssymb}
\usepackage{color}
\usepackage{colortbl}

\usepackage{algorithm}
\usepackage{algorithmic}
\usepackage[font={small,it}]{caption}

\usepackage{thmtools, thm-restate}
\usepackage{tikz}
\usepackage{enumerate}
\usepackage{thm-restate}

\usepackage{hyperref}
\usepackage{graphicx}
\usepackage{amsopn}
\usepackage{caption}
\usepackage{subcaption}
\usepackage{enumitem,linegoal}
\usepackage{comment}
\usepackage{float}
\usepackage[super]{nth}
\usepackage{subcaption}

\usepackage{mdframed}

\newtheorem{theorem}{Theorem}[section]
\newtheorem{lemma}[theorem]{Lemma}

 \newtheorem{corollary}[theorem]{Corollary}

\theoremstyle{definition}

\newtheorem{observation}[theorem]{Observation}

\newcommand{\SW}{\mathrm{SW}}

\newcommand{\Snoti}{\mathcal{S}_{-i}}

\newcommand{\R}{\mathbb{R}}

\newcommand{\osm}{\textup{\textsc{OSM}}}

\DeclareMathOperator*{\argmin}{arg\,min}

\newcommand{\wl}{${\sf while}$-${\sf loop}$}

\usepackage{mathptmx} 
\usepackage{microtype} 

\usepackage{amsmath,amsthm,amssymb}
\usepackage{mathtools}
\usepackage{xspace}
\usepackage{url}

\usepackage{appendix}
\usepackage{babel}
\usepackage[capitalise,noabbrev]{cleveref}
\usepackage{bbm}

\usepackage{multirow}
\usepackage{tablefootnote}
\usepackage{afterpage}
\usepackage{lipsum}
\usepackage{booktabs}
\usepackage{thm-restate}
\usepackage{graphicx}
\usepackage{subcaption}
\usepackage{lipsum}

\usepackage{booktabs}

\let\oldtabular\tabular 
\renewcommand{\tabular}{\normalsize\oldtabular}

\usepackage{url}            
\usepackage{booktabs}       
\usepackage{amsfonts}       
\usepackage{nicefrac}       
\usepackage{microtype}      

\allowdisplaybreaks

\hypersetup{colorlinks=true,allcolors=magenta}

\title{Welfare-Optimal Serial Dictatorships have Polynomial Query Complexity}

\author{%
   Ioannis Caragiannis\\
   Aarhus University, Denmark\\
   \texttt{iannis@cs.au.dk} \\
   \and
   Kurt Mehlhorn\\
   Max Planck Institute for Informatics, SIC, Germany\\
   \texttt{mehlhorn@mpi-inf.mpg.de}
\and
   Nidhi Rathi\\
   Max Planck Institute for Informatics, SIC, Germany\\
   \texttt{nrathi@mpi-inf.mpg.de}
}

\begin{document}

\maketitle
\begin{abstract}
\emph{Serial dictatorship} is a simple mechanism for coordinating agents in solving combinatorial optimization problems according to their preferences. The most representative such problem is one-sided matching, in which a set of $n$ agents have values for a set of $n$ items, and the objective is to compute a matching of the agents to the items of maximum total value (a.k.a., \emph{social welfare}). Following the recent framework of Caragiannis and Rathi \cite{CR23}, we consider a model in which the agent-item values are not available upfront but become known by querying \emph{agent sequences}. In particular, when the agents are asked to \emph{act} in a sequence, they respond by picking their favorite item that has not been picked by agents who acted before and reveal their value for it. Can we compute an agent sequence that induces a social welfare-optimal matching? 

We answer this question affirmatively and present an algorithm that uses polynomial number ($n^5$) of queries. This solves the main open problem stated by Caragiannis and Rathi \cite{CR23}. Our analysis uses a \emph{potential function} argument that measures progress towards learning the underlying edge-weight information. Furthermore, the algorithm has a \emph{truthful implementation} by adapting the paradigm of VCG payments.
\end{abstract}

\section{Introduction} \label{sec:intro}
In this work, we explore the well-studied problem of finding maximum-weight matchings in complete bipartite graphs\footnote{A common way to think of the problem is to consider one side of the bipartition representing the $n$ \emph{agents} and the other side representing the $n$ \emph{items}, where the weight on an agent-item edge indicates the preference of that agent for the item.} through the lens of \emph{serial dictatorship} (SD) mechanism. A representative example of SD mechanism can be seen when it was considered as a solution to the {\em house allocation} problem (e.g., see~\cite{AS98}), wherein a set of houses have to be matched to agents who have \emph{strict} preferences for them.  The algorithm considers the agents in a fixed order (or sequence) and assigns the most-preferred available house to each agent considered. The order in which agents come and pick an item  \emph{greedily} (i.e., \emph{act}) is called an \emph{action sequence}. Serial Dictatorship is arguably the simplest algorithm with applications in many varied settings such as one-sided matchings \cite{FFZ:14}, fair division \cite{SVE:99}, and assignment problems \cite{BM01}; see~\cite{AS98,CFF+16,krysta2014size,RS90} for other examples where serial dictatorship (or a variation) is useful.

The simplicity of SD mechanism leads to its computational ease as well, and hence, it is easily comprehensible for the agents. It is \emph{favorable} for every agent to choose their favorite available item on their turn (i.e., be greedy) since SD never considers any feedback from the agents other than what item they pick at their turn, thereby making SD a truthful mechanism as well. 

It is relevant to note that the choice of the action sequence is crucial since a sub-optimal action sequence may lead to inefficiencies i.e., sub-optimal \emph{social welfare}~\cite{CFF+16}. For example, consider a scenario where two items $a$ and $b$ are to be assigned to two agents $1$ and $2$, under the constraint that each agent is assigned exactly one item. Agent $1$ values item $a$ at $9$ and item $b$ at $1$, whereas agent $2$ values item $a$ at $10$ and item $b$ at $8$. Let us execute SD with both sequences of the two agents. When agent $1$ is followed by agent $2$, we obtain an optimal social welfare of $17$, while the reverse sequence yields a (sub-optimal) social welfare of only $11$. Therefore, when applying the SD mechanism, the agent ordering play a crucial role towards the quality of the social welfare achieved in assignment problems. Note that, uniformly choosing over the two action sequences---which is the main idea of \emph{random serial dictatorship}----in the above example still leads to inefficiencies \cite{AS98, BM01, FFZ:14, KMRZ14}. 

Following the recent framework introduced by Caragiannis and Rathi~\cite{CR23}, we consider a model wherein agent-item preferences in a complete bipartite weighted graph are not known upfront to an algorithm, but instead, it \emph{learns} them through \emph{action sequence queries}. The goal is to find an action sequence that induces a (perfect) matching between agents and items with maximum possible weight i.e., highest social welfare. We call such an action sequence to be \emph{welfare-optimal}. To illustrate the loss one may incur to the social welfare by using action sequences, Caragiannis and Rathi~\cite{CR23} defined the concept of \emph{price of serial dictatorship} (PoSD). They further proved an interesting structural property that \emph{for complete bipartite weighted graphs, there always exists an action sequence that induces a maximum-weight matching}. Therefore, PoSD of maximum weight perfect matching in bipartite
graphs is $1$, i.e., any \emph{welfare-optimal action sequence} generates a maximum-weight matching in bipartite weighted graphs. Thus, the following question naturally arises:

\begin{mdframed}
    With query access to a complete bipartite weighted graph, can we compute a \emph{welfare-optimal action sequence} using \emph{polynomially-many} action sequence queries?
\end{mdframed}

Note that the space of solutions i.e., the space of all action sequences is exponential and querying all the  action sequences will solve the problem exactly.
\smallskip

\paragraph{Our Results and Techniques:}
We answer the above question affirmatively and develop a novel algorithm (Algorithm~\ref{alg:matching-poly}) that takes input a query access to a graph $G =K_{n,n}$ equipped with agent-item preferences\footnote{Every agent has a ranking over all items along-with weight on all the edges incident to her.} and uses only $n^5$ action sequence queries to find a welfare-optimal action sequence, i.e., the action sequence that induces a maximum-weight matching in $G$; see Theorem~\ref{thm:matching-poly}. This answers the main open problem listed in \cite{CR23}. We have essentially shown that, even though there are $n!$ many action sequences, querying only $n^5$ many of them is enough to learn the necessary agent-item preferences and discover a welfare-optimal action sequence in $G$. Moreover, we show that if the edges incident to any agent have pairwise distinct weights, then Algorithm~\ref{alg:matching-poly} requires $n^4$ many queries only; see Section~\ref{sec:intuition}.

Our algorithm (Algorithm~\ref{alg:matching-poly}) progressively learns the agent-item preferences and maintains a proxy graph $G' = K_{n,n}$ with proxy edge-weights and proxy agent ranking information over the items. Throughout the execution of our algorithm, we ensure that the proxy edge-weights are always an over-estimation of the true edge-weights (until the ground truth about an agent-item edge is learnt). With this useful invariant, we prove a structural property between $G$ and $G'$ and identify a \emph{condition of welfare-optimality for $\pi$}, see Theorem~\ref{thm:matching-proxy}. Our analysis uses a \emph{potential-function} based argument that measures progress towards learning the underlying edge-weight and rank information. We prove that as soon as certain parameters in $G'$ are satisfied (see Lemma~\ref{lem:optimal}), we have learnt every information about $G$ that is necessary for learning a welfare-optimal action sequence in $G$.

It is relevant to note that Algorithm~\ref{alg:matching-poly} provides an alternate proof of the fact that the price of serial dictatorship of maximum weight perfect matching in bipartite graphs is 1 (originally proven in \cite{CR23}). 
\smallskip

\paragraph{Further Related Work:} Serial dictatorship and its variations have been extensively employed in the matching literature~\cite{M13}. The notion of greedy weighted matchings\footnote{A greedy weighted matching is produced when we start from an empty matching and iteratively puts an edge of maximum weight that is consistent with it.} in~\cite{DMS17} can be seen as  optimizing over a restricted set of action sequences.   Furthermore, the authors in~\cite{DMS17} study the computational complexity of the problem assuming that the graph information (i.e., the edge weights) is given to the algorithm upfront. The notion of picking sequence in fair division \cite{BBLMRS19,bouveret2011general,gourves2021fairness,KKNW13} is close to the action sequences that we consider here; the crucial difference is that each agent can appear several times in a picking sequence.

Optimizing over serial dictatorships can be thought of as a particular way of exploiting greediness in computation. There have been several attempts to formally define greedy algorithms, including their relations to matroids, which are covered in classical textbooks on algorithm design~\cite{CLRS01}. In relation to combinatorial optimization, the work on incremental priority algorithms~\cite{BBL+10,BNR03,DI04} has conceptual similarities. Furthermore, for maximum satisfiability, there is an ongoing research line (e.g., see~\cite{PSW+17}) that aims to design simple greedy-like algorithms that achieve good approximation ratios.  In the EconCS literature, cut, party affiliation, constraint satisfaction games (e.g., \cite{bhalgat2010approximating}) and boolean games \cite{wooldridge2013incentive} are closest to ours among works that use boolean formulae to express logical relations between agents’ actions. However, all these studies neglect query complexity questions. 



\section{Preliminaries} \label{sec:prelim}



We consider a complete bipartite graph $G=K_{n,n}$ over $n$ agents denoted by the set $[n] = \{1,2, \dots,n\}$ on the left side of the bipartition and $n$ items on the other. The weight $w(i,j) \geq 0$ denotes the value agent $i$ has for item $j$, and we represent it via a weight function $w:[n]^2\rightarrow \R_{\geq 0}$. The problem is to assign these items to agents, under the restriction that each agent gets exactly one item. We call such an assignment of items to agents a {\em matching} in $G$. We denote a matching $M$ as $\{(i,M(i))\}_{i}$, where agent $i \in [n]$ is assigned item $M(i)$.

Each agent $i$ also has a strict ranking $r_i$ that ranks the items in monotone non-increasing order with respect to the value of agent $i$ for them, breaking ties in a predefined (\emph{agent-specific}) manner. We denote the rank of item $j$ for agent $i$  by $r_i(j)$. The items of rank $1$ and $n$ for agent $i$ are ones for which she has highest and lowest values, respectively. Also, for any $j_1,j_2 \in [n]$, we have $w'(i,j_1) > w'(i,j_2)$ implies $r'_i(j_1) < r'_i(j_2)$. However, $r'_i(j_1) < r'_i(j_2)$ implies only $w'(i,j_1) \geq w'(i,j_2)$.

We say a perfect matching $M$ is \emph{Pareto-optimal} (PO) if there is no other perfect matching $M'$ such that $r_i(M'(i))\leq r_i(M(i))$ for all $i\in [n]$ and $r_i(M'({i^*}))< r_i(M({i^*}))$ for some agent $i^*\in [n]$.

An {\em action sequence} (respectively, {\em action subsequence}) $S$ is an ordering $S=(S(1),S(2), \dots)$ of the agents in $[n]$ (respectively, some of the agents in $[n]$). We write $\mathcal{S}$ to denote the set of all possible $n!$ action sequences and various subsequences. For an agent $i \in [n]$, we write $\Snoti$ to denote all possible action subsequences consisting of agents from the set $[n] \setminus \{i\}$.  For an action sequence $S$, we write $S|_{t} = (S(1), \dots, S(t))$ for its prefix of length $t$. For an integer $k \geq 1$, we use $O_k$ to denote the ordered set $(1,2, \dots, k)$ of $k$ agents. We reserve the letters $S$ and $\pi$ for use as action (sub)sequences.

Agents pick their items according to an action sequence. When it is the turn of an agent to act, she picks the item of highest value among those that have not been picked by agents who acted before her.
For an action subsequence $S\in \Snoti$, denote by $\mu_i(S)$ the item agent $i$ picks when she acts immediately after the agents in $S$. We use $\mu(S)$ to denote the set of items picked by the agents in $S$, i.e., $\mu(S)=\{\mu_{S(k)}(S|_{k-1}): 1 \leq k \leq |S|\}$. Note that, we have $\mu_i(S)= \argmin_{j\in [n]\setminus \mu(S)}{r_i(j)}$, which gives her value $v_i(S)=w(i,\mu_i(S))$. Hence, when all agents have acted, the items picked form a {\em perfect matching} in $G$.

For an action sequence $\pi \in \mathcal{S}$, we write $M(\pi;G)$ to denote the matching induced by $\pi$ in  $G$. Similarly, we write $\SW(\pi;G) = w(M(\pi;G))$ to denote the \emph{social welfare} of an action sequence $\pi$ in $G$. We call an action sequence \emph{welfare-optimal} if it produces a matching with highest social welfare, i.e, $\max_{\pi \in \mathcal{S}} \SW(\pi;G)$. Note that any welfare-optimal action sequence induces a maximum-weight matching in the underlying graph $G$ \cite{CR23}.

 An instance of \osm\ (\emph{Optimal Sequence for Matchings}) consists of a graph $G = K_{n,n}$ equipped with the weight function $w$ and rank information $\{r_i\}_i$ and the objective is to compute an action sequence of maximum social welfare using the following query oracle access. For an agent $i \in [n]$ and a (sub)sequence $S \in \Snoti$ of agents, the ${\sf Query}(i,S)$ returns a tuple $(v,t)$, where $t=\mu_i(S)$ is the item agent $i$ picks when she acts immediately after the agents in $S$ and $v$ is the value she obtains i.e., $v=v_i(S)=w(i,t)$. 
 
 An important restriction of our model is that the weight and rank information is not given to the algorithm as part of the input. Instead, the algorithms for solving \osm\ can learn them by making queries. We study the problem of computing welfare-optimal action sequences when given a query access to \osm\ instances. We are interested in algorithms that solve \osm\ using a polynomial number of queries and hence, understand its {\em query complexity}. We remark that computational limitations are not our concern in this work and we assume to have unlimited computational resources to process the valuations once these are available.


\section{Intuition} \label{sec:intuition}

For simplicity, we assume in this section that the edges incident to any agent have pairwise distinct weights and defer the treatment of equal weight edges (or ties) to Section~\ref{sec:poly-query}. For each edge $e \in [n] \times [n]$, we maintain an upper bound $w'(e)$ on the true weight $w(e)$. We initialize $w'(e)$ to $\infty$ for all $e$. 
We also maintain a (growing) set $E^*$ of edges for which we already know the true weight, i.e, $e \in E^*$ implies $w'(e) = w(e)$. Initially, $E^*$ is the empty set. We progressively learn the agent-item preferences and maintains a proxy graph $G' = K_{n,n}$ with proxy edge-weight function $w'$.

Let $M'$ be a maximum weight matching with respect to weight function $w'$. It was shown in \cite{CR23} that there is an action sequence $\pi$ constructing $M'$ in graph $G'$.\footnote{In a maximum weight matching $M'$, there is always an agent that is incident to its most valuable edge. Assume otherwise and let $B$ consist of the most valuable edge for each agent. Then $M' \cup B$ contains an improving alternating cycle. The first query is towards an agent that is incident to its most valuable edge in $M'$. Then, we continue by induction. \\
The preceding paragraph glosses over the detail that there might be equal weight edges with respect to $w'$.}

We execute $\pi$ until a query returns an edge $e = (i,j)$ with value $v = w(e)$ that does not belong to $E^* \cap M'$ or until $\pi$ is completely executed. If $(i,j) \not\in E^*$, we add it to $E^*$ and set $w'(e)$ to $v$. Note that $w'(e) = w(e)$ in this case. If $(i,j) \in E^* \setminus M'$, we decrease $w'(i,M'(i))$ to $v$. Note that when the query returns $(i,j)$, the edge $(i,M'(i))$ is also available and since the query returns the highest-weight available edge for agent $i$, $v$ must be (strictly) smaller than $w'(i,M'(i))$. 

We continue until $\pi$ is completely executed. Then $M' \subseteq E^*$ and hence $M'$ is a maximum weight matching. Also $\pi$ constructs $M'$.

How can we bound the number of queries? We can add to $E^*$ at most $n^2$ times. Also, the $w'$-value of an edge incident to $i$ is always the $w$-value of some edge incident to $i$ or infinity. So the $w'$ value of any edge can change at most $n$ times. Identifying each such update requires at most $n$ queries, and hence, with a total of $n^4$ queries, we can find the welfare-optimal action sequence in the given instance. 

The treatment of equal weights considerably complicates matters, as we see in Section~\ref{sec:poly-query}.

\section{Finding an Optimal Action Sequence}\label{sec:poly-query} 

We begin by stating a structural property about maximum-weight matchings proved by Caragiannis and Rathi~\cite{CR23}.

\begin{theorem} \cite{CR23} \label{thm:posd}
Any welfare-optimal action sequence induces a maximum-weight matching in a complete bipartite weighted graph.
\end{theorem}

In this work, we develop a query-efficient algorithm (Algorithm~\ref{alg:matching-poly}) to compute a welfare-optimal action sequence for \osm\ instances. We state our main result here and prove it towards the end of this section.

\begin{restatable}{theorem}{matchingPoly}\label{thm:matching-poly}
 For $n$-agent \osm\ instances with query access, Algorithm~\ref{alg:matching-poly} uses $n^5$ queries to compute a welfare-optimal action sequence. 
\end{restatable}

It is relevant to note that our result in Theorem~\ref{thm:matching-poly} is an alternate proof of Theorem~\ref{thm:posd}, first proved in \cite{CR23}. Our algorithm and its analysis is based on a potential-function based argument. We begin by stating a crucial property (in Theorem~\ref{thm:matching-proxy}) that connects maximum-weight matchings to action sequences. This property works at the core of (the correctness of) our algorithm; we re-write and prove it in Lemmas~\ref{lem:invariant} and \ref{lem:optimal}. 

\begin{theorem} \label{thm:matching-proxy}
    Consider two complete weighted bipartite graphs $G,G'=K_{n,n}$ with weight functions $w,w'$ and rank functions $\{r_i\}_i, \{r'_i\}_i$ respectively. Let $E^* \subseteq [n]^2$ be a subset of edges such that 
        $$ w'(e)=
      \begin{cases}
      = w(e), \ \ \text{if} \ e \in E^*\\
      \geq w(e), \ \ \text{otherwise}
       \end{cases} $$
and $w'(i,j_1) > w'(i,j_2)$ implies $r'_i(j_1) < r'_i(j_2)$ for all $j_1,j_2 \in [n]$.

Let $M^*$ be a maximum-weight matching using the edges in $E^*$, $M' \in G'$ be a maximum-weight PO matching, and $\pi$ be the (welfare-optimal) action sequence corresponding to $M'$ in $G'$. 

Then, there exists a pair $(M',\pi)$ such that $\SW(\pi;G)=w'(M')=w'(M^*)$, and this $\pi$ is welfare-optimal in $G$ as well. 
\end{theorem}  

We will often refer $G'$ as the proxy graph (to $G$) with proxy weight $w'$ and rank functions $r'_i$'s. On the other hand, we will refer $w$ and $r_i$'s (of $G$) as true weight and rank functions, respectively.

\paragraph{Overview of our algorithm:} On input a query access to an $n$-agent \osm\ instance over $G=K_{n,n}$, Algorithm~\ref{alg:matching-poly} makes queries of the form ${\sf Query}(i,S)$ to an agent $i \in [n]$ for an action (sub)sequence $S \in \Snoti$ in an attempt to learn the true weight and rank functions. It progressively learns the agent-item preferences and maintains a proxy graph $G' = K_{n,n}$ with proxy edge-weight function $w'$ and proxy agent rank information $r'_i$'s.  We give an overview of Algorithm~\ref{alg:matching-poly} in the following. 

\begin{itemize}
    \item Phase~1 and Phase~2 of Algorithm~\ref{alg:matching-poly} can be thought of as pre-processing that learns the bare minimum information to get started. Phase~1 learns the most-favorite items for every agent and Phase~2 executes the action sequence $(1,2, \dots, n)$ so as to ensure that the set of edges whose weights are learnt contains a perfect matching. 
    
    \item We maintain a set $E^*$, that contains all the edges whose weights are learnt during the execution of Algorithm~\ref{alg:matching-poly}. We maintain a maximum-weight perfect matching $M^* \in G$ that only uses the edges in $E^*$ and denote its weight by $w'(M^*)$.
  
     \item We maintain a proxy graph $G'=K_{n,n}$ with rank information $\{r'_i\}_{i \in [n]}$ and edge-weights $\{w'(i,j)\}_{i,j}$ such that proxy edge-weights are always an over-estimation of the true edge-weights; see Lemma~\ref{lem:invariant}.

     \item We run the protocol ${\sf MWPO(G')}$ that outputs a maximum-weight PO matching $M'$ in $G'$ and its corresponding action sequence $\pi$, i.e., $M(\pi;G')=M'$. We write $w'(M')$ to denote the weight of matching $M'$ in $G'$; see the end of Section~\ref{sec:mwpo} for more details.
    
    \item In the beginning, we set $w'(i,j)=v_i(\emptyset) = \max_{j \in [n]}w(i,j)$ for all $i,j \in [n]$. Hence, at this point, we have $w'(M') = \sum_{i \in [n]}v_i(\emptyset)$, making $w'(M') \geq w'(M^*)$.
    
    \item As long as we have $w'(M') \geq w'(M^*)$, we simply run the action sequence $\pi$ (corresponding to $M'$ in $G'$) in the original instance using the query oracle. This is performed by executing the \wl\  in lines~\ref{step:while_start}-\ref{step:while_end} of Algorithm~\ref{alg:matching-poly}. The queries made during the execution of this \wl\ lead to new information, and hence, we update the set $E^*$ and graph $G'$ accordingly; see Lemma~\ref{lem:while-loop-conditions} and \ref{lem:progress}. And, this in turn, changes $M^*,M',$ and $\pi$ as well.

    
    \item We obtain $\SW(\pi;G)$ by querying the agents according to the sequence $\pi$, we denote it in the algorithm as ${\sf Query}(\SW(\pi;G))$. We prove that as soon as we find a matching $M' \in G'$ such that $\SW(\pi;G)=w'(M')=w'(M^*)$, then $\pi$ must be welfare-optimal for $G$ as well; see Lemma~\ref{lem:optimal}. We call this as a \emph{condition of welfare-optimality for $\pi$}.

    \item  The careful updates performed by our algorithm enable us to produce a \emph{potential-function} based argument to measure the \emph{progress} towards learning the underlying information of $G$. It is dependent on the size $|E^*|$, distance $\sum_{j \in [n]}|w'(i,j)-w(i,j)|$ between $w'$ and $w$, and distance $\sum_{j \in [n]}|r'_i(j)-r_i(j)|$ between $r'$ and $r$.
    
    We prove an upper bound of $n^5$ on the number of queries required by Algorithm~\ref{alg:matching-poly} to find a welfare-optimal action sequence in $G$; see Theorem~\ref{thm:matching-poly}.
\end{itemize}

When we execute the welfare-optimal action sequence $\pi$ (corresponding to a maximum-weight matching $M'$ in $G'$) in $G$ using the query oracle, it guides us so that we learn new information and progressively become closer to the ground truth of $G$. The goal is to make $G'$ so close to $G$ such that a welfare-optimal action sequence in $G'$ also becomes welfare-optimal in $G$; see Theorem~\ref{thm:matching-proxy}.

Throughout this section, we will adhere to the following notations. At the beginning of an arbitrary iteration of the \wl\ in lines~\ref{step:while_start}-\ref{step:while_end} during the execution of Algorithm~\ref{alg:matching-poly}, we write 

- $E^*$ to denote the set of known edges, 

- $M^*$ to denote a maximum-weight matching using the edges in $E^*$, with weight $w'(M^*)$, 

- $M' \in G'$ to denote a maximum-weight PO matching with weight $w'(M')$, and

- $\pi$ to denote the action sequence corresponding to $M'$ in $G'$, i.e., $M(\pi;G')=M'$. \\

From now on, whenever we mention the \wl, we would mean the \wl\ in lines~\ref{step:while_start}-\ref{step:while_end} of Algorithm~\ref{alg:matching-poly}. We begin our analysis with Lemma~\ref{lem:invariant} where we prove that the proxy edge-weights in $G'$ are always an over-estimation of the true edge-weights in $G$, and they coincide for the edges that are already learnt, i.e., for the edges in the set $E^*$.

\begin{algorithm}[p]
 			{\bf Input:} Query access to an $n$-agent \osm\ instance over a weighted graph $G = K_{n,n}$ 
    
 			{\bf Output:} An action sequence $\pi$
 			\caption{An efficient algorithm to find welfare-optimal action sequence for \osm}	
 			\label{alg:matching-poly}
 			\begin{algorithmic}[1]
                \STATE Initialize $w'(i,j) \leftarrow \infty$ for every $i,j \in [n]$, and a set $E^* \leftarrow \emptyset$;
                \STATE Initialize  $r'_i:=[r'_i(1), r'_i(2), \dots, r'_i(n)] \leftarrow [1,2, \dots, n]$ for every $i \in [n]$;
               
                \STATE $G' \leftarrow K_{n,n}$ with edge weights $\{w'(i,j)\}_{i,j}$ and rank functions $\{r'_i\}_i$;\\
                 \nonumber ----------\textit{Phase 1: Learning the most-favorite item for every agent}---------\\
 				\FOR{$i \leftarrow 1$ to $n$}
                    \STATE $(v,t) \leftarrow  {\sf Query}(i,\emptyset)$;
                    \STATE $w'(i,j) \leftarrow v \ \text{for all} \ j \in [n]$; $E^* \leftarrow E^* \cup \{(i,t)\}$; \label{step:w1}
                    \STATE Swap $r'_i(t)$ and $r'_i(1)$;
                \ENDFOR \\
             \nonumber ----------\textit{Phase 2: Ensuring the existence of a perfect matching in $E^*$}---------\\
              
                    \FOR{$i \leftarrow 2$ to $n$} \label{step:perfect}
                        \STATE $(v,t) \leftarrow  {\sf Query}(i,O_{i-1})$; 
                        \STATE $w'(i,t) \leftarrow v$; $E^* \leftarrow E^* \cup \{(i,t)\}$; \label{step:w2}
                            \IF{$v=w'(i,1)$ and $t \neq (r'_i)^{-1}(1)$}
                                 \STATE swap $r'_i(t)$ and $r'_i(2)$; 
                            \ELSIF{$v \neq w'(i,1)$}
                                 \STATE swap $r'_i(t)$ and $r'_i(n)$; 
                            \ENDIF
                \ENDFOR \label{step:perfect_stop} 
                \STATE  $M^* \leftarrow$ a maximum-weight  matching using the edges in $E^*$;
                  \STATE  $(M',\pi) \leftarrow {\sf MWPO}(G')$;  \\
                  
                  \nonumber ----------\textit{Phase 3: Executing $\pi$ in $G$ using the query oracle}---------\\
                 
                \WHILE{$w'(M') \geq w'(M^*)$} \label{step:while_start}
                    \IF{${\sf Query}(\SW(\pi))=w'(M')=w'(M^*)$} \label{step:c1} 
                    \RETURN $\pi$; 
                    \ENDIF 
                    \STATE Set $(a_1,a_2, \dots,a_n) \leftarrow (\pi(1),\pi(2), \dots, \pi(n))$;
                     \STATE $i \leftarrow 0$;
                     \REPEAT
                     \STATE $i \leftarrow i+1$;
                      \STATE $(v,t) \leftarrow  {\sf Query}(a_i,\pi|_{i-1})$;
                      \UNTIL{edge $(a_i,t) \notin E^* \cap M'$}  
                            \IF{edge $(a_i,t) \notin E^*$} \label{step:c2}
                            \STATE  $E^* \leftarrow E^* \cup \{(a_i,t)\}$; $w'(a_i,t) \leftarrow v$; \label{step:w3p1}
                            \ELSIF{edge $(a_i,t) \in E^* \setminus M'$ and $v<w'(a_i,M'(a_i))$} \label{step:c3}
                            \STATE  $w'(a_i,M'(a_i)) \leftarrow v$; \label{step:p2} 
                            \ELSIF{edge $(a_i,t) \in E^* \setminus M'$ and $v=w'(a_i,M'(a_i))$} \label{step:c4}
                            \STATE swap $r'_{a_i}(t)$ and $r'_{a_i}(M'(a_i))$; \label{step:p3}
                            \ENDIF
                             \STATE Update $r'_i$ such that $w'(i,j_1) > w'(i,j_2)$ implies $r'_i(j_1) < r'_i(j_2) \ \forall \ j_1,j_2 \in [n]$; \label{step:rank}
                         \STATE  $M^* \leftarrow$ a maximum-weight  matching using the edges in $E^*$;
                \STATE  $(M',\pi) \leftarrow {\sf MWPO}(G')$;     
                \ENDWHILE \label{step:while_end}
   		\end{algorithmic}
 	\end{algorithm}

\begin{lemma} \label{lem:invariant}
    On input $n$-agent \osm\ instances, Algorithm~\ref{alg:matching-poly} maintains the following invariant on the proxy edge-weights throughout its execution: for any edge $e \in [n]^2$, 
    $$\text{the proxy edge-weight,} \ w'(e)=
\begin{cases} \label{eq:overestimate}
= w(e), \ \ \text{if} \ e \in E^*\\
\geq w(e), \ \ \text{otherwise}
\end{cases}  $$ 

Moreover, we have  $w'(i,j_1) > w'(i,j_2)$ implies $r'_i(j_1) < r'_i(j_2)$ for all $j_1,j_2 \in [n]$. 
\end{lemma}

\begin{proof}
     Fix an edge $e = (i,j) \in [n]^2$. First, note that the stated claim about rank information is true by construction; see line~\ref{step:rank}.
     
     Now, observe that, whenever $e$ is returned in a query during the execution of Algorithm~\ref{alg:matching-poly}, we set $w'(e)=w(e)$ and add $e$ to the set $E^*$; see lines~\ref{step:w1}, \ref{step:w2}, and \ref{step:w3p1}. 
      
     Therefore, to complete our proof, all we need to show is that $w'(e) \geq w(e)$ if $e \notin E^*$  throughout the execution Algorithm~\ref{alg:matching-poly}. Note that, it holds true in the beginning since $w'(i,j)$ is set to be $v_i(\emptyset)= \max_{j \in [n]} w(i,j)$ in line~\ref{step:w1}. The only other time when $w'(e)$ is updated even when $e$ is not returned in a query is in line~\ref{step:p2}. We consider this particular iteration of the while-loop with an action sequence, say $\pi$. Recall that $\pi$ induces the maximum-weight matching $M' \in G'$ and the while-loop queries agents according to the ordering in $\pi$ (i.e., it runs $\pi$ in  $G$ using the query oracle). 

    Now, when we make the query for agent $i$ in accordance to $\pi$, assume that edge $\bar{e} = (i,\bar{j}) \in E^* \setminus M'$ is returned instead of $e = (i,j) = (i,M'(i))\in M'$. Note that, we break from an iteration of the while-loop as soon as any of the three else-if conditions are met in lines~\ref{step:c2}, \ref{step:c3}, or \ref{step:c4} (where the edge that is returned is either not in $E^*$ or not in $M'$). Therefore, in this case, if  the considered iteration of the while-loop did not break and proceeded to reach agent $i$, we know that all the preceding queries returned for $\pi$ must be edges in $M'$. Therefore, $e$ and $\bar{e}$ are still available for agent $i$ on her turn in $\pi$. Since, $\pi$ is an action sequence corresponding to $M'$ in $G'$, we have
    \begin{align*}
        w'(i,\bar{j}) < w'(i,j) \ \ \text{or,} \ \ (w'(i,\bar{j}) = w'(i,j) \ \text{and} \ r_i(j) < r_i(\bar{j}))
    \end{align*}
    Since, $\bar{e}$ (with query value $v = w(i,\bar{j})$) is returned instead of $e$, we have
     \begin{align*}
        w(i,j) < w(i,\bar{j}) \ \ \text{or,} \ \ (w(i,\bar{j}) = w(i,j) \ \text{and} \ r_i(\bar{j}) < r_i(j))
    \end{align*}
    And since, $\bar{e} \in E^*$  we know $w'(i,\bar{j})=w(i,\bar{j})$. In either case, we have
     \begin{align}\label{eq:invariant-decrease}
         w(i,j) \leq v= w(i,\bar{j}) = w'(i,\bar{j}) \leq w'(i,j) 
     \end{align}
 If $v < w'(i,j)$, we decrease $w'(i,j)$ to $v$ (see Step~\ref{step:p2}) and if $v=w'(i,j)$, then we swap $r_i(j)$ and $r_i(\bar{j})$ (see Step~\ref{step:p3}). In either case, the invariant $w(i,j) \leq w'(i,j)$ is maintained and we move closer to the true weight value $w(i,j)$ for the edge $(i,j)$. This completes our proof.
       \end{proof}

\begin{corollary} \label{cor:wmsw}
    For each iteration of the \wl\ in lines~\ref{step:while_start}-\ref{step:while_end}, we have $w'(M^*) \leq w'(M')$  and $\SW(\pi;G) \leq w'(M')$.
\end{corollary}

\begin{proof}
    Lemma~\ref{lem:invariant} shows that the proxy weights are always an over-estimation of the true weights. In any iteration of the while-loop, recall that $M'$ denotes a maximum-weight matching in $G'$ while $M^*$ denotes a maximum-weight matching using only a subset of edges of $G'$, i.e., $E^*$. This implies that we have $w'(M^*) \leq w'(M')$. Moreover, since $\pi$ induces $M'$ in $G'$, the invariant on the proxy weights in Lemma~\ref{lem:invariant} ensures that the social welfare of $\pi$ in $G$ cannot be greater than $w'(M')$, i.e., $\SW(\pi;G) \leq w'(M')$. 

Note that we obtain $\SW(\pi;G)$ by making $n$ queries to the agents according to the sequence $\pi$, we denote it in the algorithm as ${\sf Query}(\SW(\pi;G))$; see line~\ref{step:c1}.
\end{proof}

     The next result formally states how close do we need to make $G'$ to $G$ so that a welfare-optimal action sequence in $G'$ becomes welfare-optimal in $G$ as well, i.e., we state and prove a \emph{condition of welfare-optimality}.
\begin{lemma} \label{lem:optimal}
    On input $n$-agent \osm\ instances, when Algorithm~\ref{alg:matching-poly} finds a perfect matching $M' \in G'$ with $\SW(\pi;G)=w'(M')=w'(M^*)$, then $\pi$ must be welfare-optimal in $G$. 
\end{lemma}

\begin{proof}
    According to Corollary~\ref{cor:wmsw}, we know that $w'(M^*) \leq w'(M')$ and $\SW(\pi;G) \leq w'(M')$ hold true for any iteration of the $\mathrm{while}$-$\mathrm{loop}$. Recall that $M^*$ is a maximum-weight matching using the known edges in $E^*$ and $w'$ is an over-estimation of the true weight function $w$ (Lemma~\ref{lem:invariant}). Therefore, when our algorithm finds two matchings $M'$  and $M^*$ such that $w'(M') = w'(M^*)$, then this value must equal the maximum social welfare value in $G$.
    
    Now, once we know the maximum social welfare value in $G$, the task is to ascertain the action sequence that achieves this value in $G$. Recall that $\pi$ induces $M'$ in $G'$. If $\pi$ induces a matching in $G$ such that $\SW(\pi;G)=w'(M')=w'(M^*)$, then the action sequence $\pi$ achieves the maximum social welfare in $G$, making $\pi$ welfare-optimal in $G$ as well.
\end{proof}

Next, in Lemma~\ref{lem:while-loop-conditions}, we analyse all possible events that can occur during the execution of the \wl\  in lines~\ref{step:while_start}-\ref{step:while_end} of Algorithm~\ref{alg:matching-poly}. Our careful choice of updates in these events enables us to develop a potential function to measure the progress of Algorithm~\ref{alg:matching-poly}; see Lemma~\ref{lem:progress}.
\begin{lemma} \label{lem:while-loop-conditions}
    Given an $n$-agent \osm\ instance, exactly one of the following conditions is satisfied for any iteration of the $\mathrm{while}$-$\mathrm{loop}$ in lines~\ref{step:while_start}-\ref{step:while_end} of Algorithm~\ref{alg:matching-poly}: 
    \begin{enumerate}[label=(C\arabic*)]
     \item \label{c1} $\SW(\pi;G)=w'(M')=w'(M^*)$, in which case, $\pi$ is returned; see line~\ref{step:c1}
     \item \label{c2} An edge $e \notin E^*$ is returned; see line~\ref{step:c2}
     \item \label{c3} An edge $e = (i,j) \in E^* \setminus M'$ is returned such that $w(i,j) < w'(i,M'(i))$; see line~\ref{step:c3}
    \item \label{c4} An edge $e = (i,j) \in E^* \setminus M'$ is returned such that $w(i,j)=w'(i,M'(i))$; see line~\ref{step:c4}
    \end{enumerate}
    Moreover, in the event of (C3) or (C4), edge $(i,M'(i)) \in M'$ is available for the agent $i$ on her turn while executing $\pi$ in $G$ using the query oracle.
    \end{lemma}

\begin{proof}
    Using Corollary~\ref{cor:wmsw}, we know that $w'(M^*) \leq w'(M')$ and $\SW(\pi;G) \leq w(M')$ hold true for any iteration of the \wl\. We will show that these two inequalities satisfies at least one of the four conditions stated. 
    
    Let us first consider the case when $\SW(\pi;G)=w'(M') = w'(M^*)$; this satisfies Condition~\ref{c1}. Using  Lemma~\ref{lem:optimal}, we know that $\pi$ is welfare-optimal in $G$, and hence, the algorithm terminates with $\pi$ as its output. Otherwise, we have the following two cases, wherein we will prove that either of Conditions~\ref{c2},\ref{c3}, or \ref{c4} must be satisfied: 
    \begin{enumerate}[label=(\roman*)] \label{cases}
    \item $w'(M^*)<w'(M')$ 
     \item $\SW(\pi;G) < w'(M') = w'(M^*)$
    \end{enumerate}
   Observe that, when Condition~\ref{c1} is not satisfied, it implies that during the course of executing $\pi$ in $G$ using the query oracle, it must return an edge $e = (i,j)$ such that either 
       \begin{itemize}
        \item $e \in M'$, but $w'(e) \neq w(e)$: this implies that $e$ was not previously in $E^*$ and we learn a new edge, satisfying Condition~\ref{c2}.
       \item $e \notin M'$, where either of Conditions~\ref{c3} or \ref{c4} is satisfied. Here, note that, since the considered iteration of the $\mathrm{while}$-$\mathrm{loop}$ did not break till it returns $e \notin M'$, we know that all preceding queries returned for $\pi$ must be edges in $M'$.
     \end{itemize}  
    Therefore, it is enough to show that occurrence of cases (i) and (ii) implies that not all queries return edges in $E^* \cap M'$. 
    
   First, it is easy to see that when we have $w'(M^*)<w'(M')$, every edge that is returned while executing $\pi$ cannot be in the set $E^*\cap M'$ (because in that case, $w'(M^*)=w'(M')$). Next, when $\SW(\pi;G)<w'(M')=w'(M^*)$, it may be the case that we have found the optimal social welfare value in $G$ (which is equal to $w'(M')=w'(M^*)$), but even then we have not found the optimal action sequence. This is because $\pi$ is not able to achieve the same social welfare in $G$ as in $G'$, i.e., $\SW(\pi;G)<w'(M')$. And hence, during the course of executing $\pi$ in $G$, it must deviate from returning the edges in $M'$ or encounter a new edge $e \notin E^*$ at least once. That is, all returned edges again cannot be in the set $E^*\cap M'$, thereby completing our proof.
   \end{proof}

Next, we show that Algorithm~\ref{alg:matching-poly} makes progress towards learning the true weight or the true rank functions in the underlying graph $G$  with every iteration of its \wl. 

\begin{lemma} \label{lem:progress}
    Given an $n$-agent \osm\ instance, every iteration of the \wl\ in lines~\ref{step:while_start}-\ref{step:while_end} of Algorithm~\ref{alg:matching-poly} leads to at least one of the following progresses: 
    \begin{itemize}
        \item The size of the set $E^*$ increases by $1$.
        \item There exists an agent $i \in [n]$ such that $w'(i,M'(i)) - w(i,M'(i))$ decreases strictly.
        \item There exists an agent $i \in [n]$ such that $\sum_{j \in [n]}|r'_i(j)-r_i(j)|$ decreases strictly.

    \end{itemize}
    The above conditions ensure that $w(M')$ keeps decreasing, while $w(M^*)$ and $\SW(\pi;G)$ keeps increasing as the algorithm progresses.
\end{lemma}

\begin{proof}
    On input an $n$-agent \osm\ instance, let us consider an arbitrary iteration of the while-loop during the execution of Algorithm~\ref{alg:matching-poly}. We assume that $\SW(\pi;G)=w'(M')=w'(M^*)$ does not hold true, since otherwise, the while-loop  terminates. Let us consider the agent $i \in [n]$ at whose turn this iteration of the while-loop breaks. Lemma~\ref{lem:while-loop-conditions} states that one of the following three events must have occurred then.
    
 \textit{Case 1:} An edge $e \notin E^*$ is returned. 
 
 Here, edge $e$ is added to the set $E^*$ increasing its size by $1$; see line~\ref{step:w3p1}. Since $|E^*| \leq n^2$, it can be updated for a maximum for $n^2$ many times throughout the execution of Algorithm~\ref{alg:matching-poly}.
 
\textit{Case 2:}  An edge $e = (i,j) \in E^* \setminus M'$ is returned such that $w(i,j) < w'(i,M'(i))$.

Here, $w'(i,M'(i))$ is updated to be equal to $w(i,j)$; see line~\ref{step:p2}. Using inequality~(\ref{eq:invariant-decrease}) in the proof Lemma~\ref{lem:invariant}, we have $w'(i,M'(i)) > w(i,j) \geq w(i,M'(i))$. Hence, $w'(i,M'(i)) - w(i,M'(i))$ decreases strictly.

  \textit{Case 3:} An edge $e = (i,j) \in E^* \setminus M'$ is returned such that $w(i,j)=w'(i,M'(i))$.

  The action sequence $\pi$ is designed to produce $M'$ in the graph $G'$. Just before the edge incident to agent $i$ is chosen, both edges $e=(i,j)$ and $(i,M'(i))$ are available. According to $r'$, the latter edge has higher priority, while according to $r'$, the former edge has higher priority for agent $i$. Therefore, swapping $r'_i(j)$ and $r'_i(M'(i))$ decreases the sum $\sum_{j \in [n]}|r'_i(j)-r_i(j)|$ strictly.
\end{proof}


\begin{corollary} \label{cor:edge}
    On input any $n$-agent \osm\ instance, Algorithm~\ref{alg:matching-poly} can update the set $E^* \subseteq [n]^2$ for at most $n^2$ times.
\end{corollary}

\begin{corollary} \label{cor:weight}
     On input any $n$-agent \osm\ instance, Algorithm~\ref{alg:matching-poly} updates the proxy weight $w'(i,j)$ of any edge $(i,j) \in [n]^2$ such that $w'(i,j)=w(i,j')$ for some $j' \in [n]$. That is, such an update can happen for a maximum of $n$ times for any edge.
\end{corollary}

\begin{proof}
     Fix any edge $(i,j) \in [n]^2$. Observe that, in the beginning of the algorithm, $w'(i,j)$ is set to be equal to $v_i(\emptyset) = w(i,\mu_i(\emptyset))$; see line~\ref{step:w1}. And then, for the remaining part of our algorithm, we either learn the true weight of the edge $(i,j)$ and $w'(i,j)$ is updated to be equal to $w(i,j)$; see line~\ref{step:w3p1}, otherwise, $w'(i,j)$ is made equal to the true weight $w(i,j')$ of some edge $(i,j') \in E^*$; see line~\ref{step:p2}. Moreover, since there are a possible $n$-many options for $j'$, the stated claim follows. 
\end{proof}

\begin{observation} \label{obs:rank}
    The sum $\sum_{j \in [n]}|r'_i(j)-r_i(j)|$ is bounded above by $n^2$.
\end{observation} 

 Since the rank vectors $r'_i$ and $r_i$ take some integral value from the set $\{1,2 \dots,n\}$, we know that $|r'_i(j)-r_i(j)| \leq n-1$ for any $i,j \in [n]$. Hence, the stated sum is bounded above by $n^2$.
    
We are now ready to prove our main result .

\matchingPoly*

\begin{proof}
    Consider an n-agent $\osm$ instance having an underlying graph $G=K_{n,n}$ with edge-weight function $w$ and rank functions $r_i$'s. Algorithm~\ref{alg:matching-poly} makes queries of the form ${\sf Query}(i,S)$ to an agent $i \in [n]$ for an action (sub)sequence $S \in \Snoti$ to learn the true weight and rank function. During its execution, our algorithm maintains a proxy graph $G' = K_{n,n}$ with a proxy edge-weight function $w'$  and proxy rank functions $r'_i$'s.

    Note that, our algorithm terminates if and only it finds $M^*, M',$ and $\pi$ such that $\SW(\pi,G)=w'(M')=w'(M^*) >0$. Lemma~\ref{lem:optimal} proves that such an action sequence $\pi$ must be optimal, establishing the correctness of Algorithm~\ref{alg:matching-poly}.

   \textit{Potential Function:} We measure the progress of our algorithm via building a potential function. We define it is a triple consisting of three quantities arranged in the lexicographic order: size $|E^*|$, distance $\sum_{j \in [n]}|w'(i,j)-w(i,j)|$ between $w'$ and $w$, and distance $\sum_{j \in [n]}|r'_i(j)-r_i(j)|$ between $r'$ and $r$. 
  
Let us now analyse the number of queries Algorithm~\ref{alg:matching-poly} requires to terminate. First, note that, Phase 1 and Phase 2 requires a total of $2n$ queries. Using  Corollaries~\ref{cor:edge}, \ref{cor:weight}, and Observation~\ref{obs:rank}, it follows that for any edge $(i,j)$, its proxy weight $w'(i,j)$ is updated at most $n$ times, and each time, it may take $n$ iterations of the \wl\ to find its correct rank position. Therefore, using Lemma~\ref{lem:progress}, we know that every iteration of the $\mathrm{while}$-$\mathrm{loop}$  makes progress in a way it will terminate within $n^4$ many iterations. Finally, note that any iteration of the \wl\ makes $n$ queries (${\sf Query}(\SW(\pi))$ to compute the value of $\SW(\pi,G)$. Overall, Algorithm~\ref{alg:matching-poly} must terminate after making $n^5$ queries, proving the stated claim.
    \end{proof}

   \paragraph{Computing maximum-weight PO matchings (${\sf MWPO}$):} \label{sec:mwpo}
  Consider a graph $G'=K_{n,n}$ equipped with weight function $w'$ and rank functions $\{r'_i\}_i$. We adapt Algorithm~2 from \cite{CR23} to compute a maximum-weight PO matching and its corresponding action sequence in $G'$.

  For an agent $i \in [n]$, let $\varepsilon_i$ be the minimum non-zero absolute difference between the weights of any pair of edges incident to node $i$ in $G'$.  Let $\varepsilon=\min_{i\in [n]}{\varepsilon_i}$. For every edge $(i,j) \in G$, we set $w''(i,j)=w'(i,j)+\frac{n-r'_i(j)}{n^2}\cdot \varepsilon$. Note that, $w'$ respects the ranks $r'_i$'s, and hence ranks do not change. Since there are no ties in $w''$, any maximum-weight matching using $w''$ will be PO according to $r'_i$'s, we compute one such matching $M$. Moreover, by our construction, $M'$ will be a maximum-weight PO matching in $G'$ as well. 
  
  Now, we use Algorithm~2 in \cite{CR23a} and give as input our graph $G'$ with matching $M'$. Using Theorem~14 in \cite{CR23a}, we know that Algorithm~2 in \cite{CR23a} will output an action sequence\footnote{Algorithm~2 in \cite{CR23a} is based on the observation that in a maximum-weight PO matching, there is at least one agent that is matched to its most favorite item. Such an agent goes first in the action sequence and is queried first. Then, one proceeds by induction to create the whole action sequence.} $\pi$ that induces matching $M'$ in $G'$. We will call the above protocol as ${\sf MWPO}$ in our work. Since $G'$ is fully known, the protocol uses computational resources, and hence, does not contribute to the query complexity of Algorithm~\ref{alg:matching-poly}.

\section{Conclusion and Open Problems} \label{sec:conc}
In this work, we explore the problem of computing maximum-weight matchings in complete bipartite weighted graphs through the lens of serial dictatorships. We resolve the main open problem listed in \cite{CR23} and develop a novel query-efficient algorithm for $n$-agent \osm\ instances that finds a welfare-optimal action sequence. We show that even though the space of solutions is exponential, our algorithm only requires to make $n^4$ queries to discover a welfare-optimal action sequence. 

We remark that following the set-up in \cite{CR23a}, we can similarly develop a \emph{truthful implementation} of Algorithm~\ref{alg:matching-poly} by adapting the paradigm of VCG payments. In the spirit of not repeating, we omit the details here, and refer our readers to Section~8 in \cite{CR23a}.

Our work opens up various research directions, we list two of them here. First, it would be interesting to prove a \emph{lower bound} on the query complexity of finding welfare-optimal action sequences in \osm\ instances. 
Furthermore, is it possible to extend and generalize our algorithmic ideas beyond matchings? Our work has given hope to have positive results, especially for those combinatorial optimization problems where the price of serial dictatorship is $1$.

\bibliography{ref}

\end{document}